\documentclass{article}
\usepackage{arxiv}
\newcommand{\emps}{\epsilon}
\renewcommand{\mathtt}[1]{\mbox{\texttt{#1}}}

\usepackage{latexsym,amsmath,amssymb,color,graphicx}

\newcommand{\fitcite}{\hspace{-1pt}\hspace{1pt}}

\newenvironment{Array}[1]{\begin{array}{@{}#1@{}}}{\end{array}}

\makeatletter
\newenvironment{prog}[2][{}]{%
\def\@endOfProgMark{#1}%
\setbox0=\hbox\bgroup$\begin{Array}{#2}%
}
{\end{Array}$\egroup%
\newdimen{\tmpdimen}\setlength{\tmpdimen}{\ht0}\addtolength{\tmpdimen}{\dp0}%
\addtolength{\tmpdimen}{-\baselineskip}\begin{center}%
    \makebox[0pt]{\box0}%
    \makebox[0pt]{\hbox to \linewidth {\hfill\raisebox{-.5\tmpdimen}{\@endOfProgMark}}}%
  \end{center}%
}%
\makeatother

\newcommand{\mi}[1]{{\ifmmode{\mathit{#1}}\else{$\mathit{#1}$}\fi}}
\newcommand{\mbf}[1]{{\ifmmode{\mathbf{#1}}\else{$\mathbf{#1}$}\fi}}
\newcommand{\ms}[1]{{\ifmmode{\mathsf{#1}}\else{$\mathsf{#1}$}\fi}}
\newcommand{\mr}[1]{{\ifmmode{\mathrm{#1}}\else{$\mathrm{#1}$}\fi}}
\DeclareMathAlphabet{\mathsc}{OT1}{cmr}{m}{sc}

\newcommand{\setof}[1]{\{#1\}}
\newcommand{\scmp}[2]{\{{#1}\mid{#2}\}}

\newcommand{\ra}{\Rightarrow}

\newcommand{\any}{\raisebox{-0.3em}{$-$}}

\newcommand{\typeofOp}{:}
\newcommand{\typeof}[2]{{#1}\mathbin{\typeofOp}{#2}}

\renewcommand{\any}{\centerdot}

\newcommand{\shorteq}{\mbox{\raisebox{1pt}{--}\hspace{-0.5em}\raisebox{-0.5pt}{--}}}

\newcommand{\qed}{\hfill\mbox{$\Box$}}

\usepackage{cite,doi}
\usepackage[all]{xy}
\newtheorem{theorem}{Theorem}
\newtheorem{corollary}[theorem]{Corollary}
\newtheorem{lemma}[theorem]{Lemma}
\newtheorem{definition}[theorem]{Definition}
\newtheorem{proof}{Proof Sketch}

\renewcommand{\prog}[3]{%
  \setbox0=\hbox{$\begin{array}{@{}#1@{}}#2\end{array}$}%
  \newdimen{\tmpdimen}\setlength{\tmpdimen}{\ht0}\addtolength{\tmpdimen}{\dp0}
  \addtolength{\tmpdimen}{-\baselineskip}\[
  \hbox to \linewidth {
    \hfill\makebox[0pt][c]{$\begin{array}{@{}#1@{}}#2\end{array}$}\hfill
    \makebox[0pt][r]{\raisebox{-.5\tmpdimen}{#3}
    }
  }\]
}

\title{Translation of Regular Expression with Lookahead into Finite State Automaton}
\author{Akimasa Morihata}

\begin{document}

\maketitle

\begin{abstract}
Most of the conventional implementations of regular expressions are based on backtracking. Such implementations are slow in the worst case, and thus,  
we would like to develop a better matching algorithm. 
However, it is nontrivial to provide an efficient matching algorithm
that can deal with practical extensions including submatch addressing.
This paper studies regular expression with lookaheads and negative lookaheads, abbreviated to REwLA. 
First, we propose a transformation from a REwLA of size $m$ to a deterministic finite automaton of $\mr{O}(2^{2^m})$ states.
Next, we consider weighted regular expressions, which enable us to calculate submatch addressing. We propose a transformation from a weighted REwLA of size $m$ to 
a weighted nondeterministic finite automaton of $\mr{O}(2^{2^m})$ states.
\end{abstract}

\thispagestyle{empty}

\newcommand{\pla}[1]{{(?{\shorteq}\,{#1})}}
\newcommand{\nla}[1]{{(?!\,{#1})}}
\newcommand{\rept}{{*}\,}

\section{Introduction}
Regular expressions have been studied for a long time, and many interesting results are known, including their correspondence to finite state automata.
At the same time, regular expressions are widely used in practice, notably in lexical analyzers and scripting languages.
However, the regular expressions used in, for example, Perl and Ruby differ from the theoretically well-known ones in various respects.

One of the largest differences is that, although it is well known that converting a regular expression into an NFA (Nondeterministic Finite Automaton) or a DFA (Deterministic Finite Automaton) enables efficient regular expression matching (i.e., checking whether a string is contained in the language denoted by the expression), in practice most regular expression processing systems use backtracking.
For this reason, in current Ruby\footnote{We conducted the experiment with Ruby 1.8.7.}, for example,
the following matching requires time proportional to the exponential of the string length.
\[\mbox{``}\mathtt{a}\mathtt{a}\cdots \mathtt{a}\mathtt{b}\mbox{''} \mathbin{=\sim}
/((\mathtt{a}*)\rept\mathtt{b})\rept\mathtt{b}/\]

{\def\fst{\mathtt{'}\backslash\mathtt{1'}}

Many processing systems use backtracking mainly because of the extraction of submatches.
Consider $\mathtt{b} \mid (\mathtt{a}\mid \mathtt{b})\rept\mathtt{bb}$,
which denotes the same language as $((\mathtt{a}\rept)\rept\mathtt{b})\rept\mathtt{b}$. While Ruby can perform matching efficiently for $\mathtt{b} \mid (\mathtt{a}\mid \mathtt{b})\rept\mathtt{bb}$, we cannot substitute it for $((\mathtt{a}\rept)\rept\mathtt{b})\rept\mathtt{b}$, because their semantics are different: the results of performing $\mathtt{str}.\mathtt{sub}(/((\mathtt{a}\rept)\rept\mathtt{b})\rept\mathtt{b}/,\fst)$ and
$\mathtt{str}.\mathtt{sub}(/\mathtt{b} \mid (\mathtt{a}\mid \mathtt{b})\rept\mathtt{bb}/,\fst)$ are in general clearly different, where $\mathtt{sub}$ is the method of replacing the matched substring and $\fst$} is the reference to the string captured by the first parenthesis.
This implementation prohibits applying many of the well-known techniques for regular expressions. This problem has been recognized relatively recently, and
techniques that compute matching without backtracking have been proposed and implemented \cite{DuFe00,Laur01,Cox07,FiHW10}.

Another major difference is that the real-world regular expressions have many extensions.
For example, Ruby's regular expressions include character classes, matching against the beginning of a line or a word, interval quantifiers, control of greedy and non-greedy matching, back-references, lookaheads, and so on.
Many of these extensions have been studied, but there are two representative extensions that existing methods cannot handle: back-references and lookaheads.

In this paper, we study the matching of regular expressions with lookaheads and negative lookaheads (Regular Expression with (negative) Lookaheads, abbreviated below as {\em REwLA}).
A lookahead checks whether the string from that position matches a specified regular expression without fetching the next character (i.e., moving the position forward).
Following Perl, Ruby, and others, we denote the positive and negative lookaheads by regular expression $e$ as $\pla{e}$ and $\nla{e}$, respectively.
For example, the following REwLA matches a string containing all of $\mathtt{a}$, $\mathtt{b}$, and $\mathtt{c}$.
\[\pla{{\any}*\mathtt{a}}\pla{{\any}*\mathtt{b}}\pla{{\any}*\mathtt{c}}\any\rept\]
Here, $\any$ denotes the pattern matching any single character.
Thus, $\pla{{\any}*\mathtt{a}}$ requires that $\mathtt{a}$ occurs somewhere in the string.
Writing this without using lookaheads is cumbersome, because the possible orderings of $\mathtt{a}$, $\mathtt{b}$, and $\mathtt{c}$ must be taken into account. For another example, the following REwLA matches a comment in the C language.
Note that $\mathtt{*}$ here is a character, not the $*$ operator of regular expressions.
\[\mathtt{/*} (\nla{\mathtt{*/}}\any)\rept \mathtt{*/}\]
$\nla{\mathtt{*/}}$ is a negative lookahead that matches strings not starting from $\mathtt{*/}$.
Thus, $(\nla{\mathtt{*/}}\any)\rept$
matches any string that contains no $\mathtt{*/}$.
Using regular expressions without negative lookaheads, this example would become fairly complicated and a potential source of bugs.
As seen above, REwLAs are useful, and an efficient matching algorithm is desirable.

Nevertheless, to the best of the author's knowledge, there has been almost no formal discussion of REwLA.
If efficiency is disregarded, matching with REwLA is easy.
One can simply check, at each position, whether a lookahead succeeds.
This naive method has a worst-case running time quadratic in the string length.
For processing large text data, it can hardly be said to be practical.

In this paper, we show a translation from REwLA to automata, and thereby provide a matching algorithm.
This approach is desirable for the following three reasons.
\begin{itemize}
\item It is based on a well-studied theory. This allows existing results to be applied for improvements and extensions.
\item The running time is proportional to the input string length, and the space complexity is independent of the input string length; hence, even very long strings can be handled efficiently.
\item The string is scanned only once from the beginning to the end.
This is suitable for online processing, such as the analysis of packets flowing through a network.
\end{itemize}

One may notice the similarity between lookahead and the intersection operation. Since the intersection of regular languages is regular,
it is natural to expect that REwLA easily corresponds to automata.
However, the construction of automata from REwLA is nontrivial.
For instance, consider the following REwLA.
\[(\pla{\mathtt{aa}}\mathtt{a})\rept\mathtt{a}\]
This REwLA matches a sequence of one or more $\mathtt{a}$'s. However,
its subexpression $\pla{\mathtt{aa}}\mathtt{a}$ matches only the empty string.
This example indicates that it seems impossible to construct the automaton corresponding to a REwLA in the same method as for the ordinary regular expressions.
In fact, for REwLA, nobody has explicitly proved that the language it denotes is regular, let alone its efficient matching algorithm.

In Section \ref{section:DFA}, we show a method to construct a DFA from a REwLA.
The idea is to convert a REwLA into a Boolean Finite Automaton \cite{BrLe80} (abbreviated to BFA).
This method provides an efficient matching algorithm for the case where submatches are not needed, and at the same time, serves as a proof that the language denoted by a REwLA is regular.
The construction obtains, from a REwLA of length $m$, a DFA with $\mr{O}(2^{2^m})$ states.

Then, in Section \ref{section:weight}, we consider weighted REwLA.
Weighted regular expressions \cite{DrKV09} enable us to express various computations including the extraction of submatches.
We show the translation from weighted REwLA to weighted NFA
by combining the method of the previous section with the standard weighted NFA construction method.
The number of states of the weighted NFA obtained from a REwLA of length $m$ is $\mr{O}(2^{2^m})$, and
therefore the matching for an input string of length $n$ can be computed in $\mr{O}(n \cdot 2^{2^m})$ time.

\renewcommand{\pla}[1]{{\langle?{\shorteq}\,{#1}\rangle}}
\renewcommand{\nla}[1]{{\langle?!\,{#1}\rangle}}

 \section{Regular Expressions with Positive and Negative Lookaheads}
\subsection{Preliminary}
$|S|$ denotes the number of elements in the set $S$.
A function from set $A$ to set $B$, denoted by $\typeof{f}{A \to B}$,
is a subset of $A \times B$ satisfying
$((a,b) \in f \wedge (a,b') \in f) \ra b = b'$.
$\setof{e_1 \mid e_2}$ is the minimum set satisfying  $\forall x_1,\ldots,x_n.\, (e_2 \ra (e_1 \in \setof{e_1 \mid e_2}))$, where $x_1,\ldots,x_n$ are free variables in $e_2$ independent from $e_1$.
For example, given functions $f$ and $g$,
$\setof{(a,c)\mid(a,b) \in f \wedge (b,c)\in g}$ is the function composition $g \circ f$.
For associative-commutative binary operator $\oplus$ with the unit on set $S$,
$\bigoplus S$ is a shorthand of $\bigoplus_{x \in S} x$.

$\mathcal{F}_{X}$ is the set of propositional formulas over propositional variables $X$.
Note that
$\mathcal{F}_X$ consists of $\mr{O}(2^{2^{|X|}})$ formulas up to equivalence.
Given formula $f \in \mathcal{F}_{X}$ and propositional logical variables $T \subseteq X$,
$T \vdash f$ denotes the fact that $f$ is true if each $x \in T$ is true and each $x' \in X\setminus T$ is false.

This paper considers a finite alphabet $\Sigma$.
$\sigma$ is a metavariable ranging over $\Sigma$.
$\Sigma^*$ is the set of strings over $\Sigma$.
The empty string and the string concatenation are denoted by $\emps$ and $\cdot$, respectively.
These notations are also used for regular expressions. We omit $\cdot$ unless it causes confusion.

Unless otherwise stated, $n$ and $m$ are the lengths of the input string and the regular expression, respectively.

Basic knowledge of regular expressions and finite state automata is assumed.
Refer to textbooks such as \cite{HoMU06} if necessary.
The functional language Haskell \cite{haskell98} is used to express computations.

\subsection{Regular Expressions with Positive and Negative Lookaheads}
We consider \emph{Regular Expressions with Positive and Negative Lookaheads} (abbreviated to REwLA), whose syntax is defined below.
\[\begin{array}{@{}l@{~~}c@{~~}l@{}}
e &::=& 
\emps
~\,{\mid}~\, \sigma
~\,{\mid}~~ e \mid e
~~{\mid}~\, e \cdot e
~\,{\mid}~\, e{*}
~\,{\mid}~\, \pla{e}
~\,{\mid}~\, \nla{e}
\end{array}\]
Here, $\sigma \in \Sigma$, and
$\pla{e}$ and $\nla{e}$ are respectively positive and negative lookaheads by REwLA $e$.
Regarding operator precedence,
${*}$ is the most tightly bound and $\mid$ is the weakest.

The semantics of REwLA is defined as follows. Note that the semantics is different from the one in Ruby.
The correspondence of these two semantics can be expressed by $(?{\shorteq}\,e) \equiv \pla{e\,{\any*}}$ and $(?!\,e) \equiv \nla{e\,{\any*}}$, where $\any$ is the pattern matching any single character in $\Sigma$.
\begin{definition}
The language defined by REwLA $e$, denoted by $\mathcal{L}(e)$, is defined as follows.
\[\begin{array}{@{}lcl@{}}
\lefteqn{w \in \mathcal{L}(e) \iff \emps \in \mathcal{L}'(e,w)}\\
\mathcal{L}'({\emps},w) &=& \setof{w}\\
\mathcal{L}'(\sigma,\emps) &=& \emptyset\\
\mathcal{L}'(\sigma,\sigma' w) &=& \mbf{if}~\sigma = \sigma'~\mbf{then}~\setof{w}~\mbf{else}~\emptyset\\
\mathcal{L}'(e \mid e',w) &=& \mathcal{L}'(e,w) \cup \mathcal{L}'(e',w)\\
\mathcal{L}'(e \cdot e',w) &=& \scmp{y}{x \in \mathcal{L}'(e,w) \wedge y \in \mathcal{L}'(e',x)}\\
\mathcal{L}'(e{*},w) &=&
\setof{w} \cup
\scmp{y}{x \in \mathcal{L}'(e,w) \wedge x \neq w \wedge y \in \mathcal{L}'(e{*},x)}\\
\mathcal{L}'(\pla{e},w) &=&
\mbf{if}~w \in \mathcal{L}(e)~\mbf{then}~\setof{w}~\mbf{else}~\emptyset\\
\mathcal{L}'(\nla{e},w) &=&
\mbf{if}~w \not\in \mathcal{L}(e)~\mbf{then}~\setof{w}~\mbf{else}~\emptyset
\qed
\end{array}\]
\end{definition}
$\mathcal{L}'$ eliminates the matched prefix and enumerates the possible suffixes.
The positive and negative lookaheads examine whether the remaining suffix matches the specified expression.
Note that either the REwLA expression or the target string grows smaller during the evaluation of $\mathcal{L}'$.

The definition immediately provides an algorithm to calculate REwLA matching; however, it is impractical because it enumerates all the possibilities and examines the remaining suffixes for every lookahead.

\subsection{Weighted Regular Expression}
Weighted regular expressions \cite{DrKV09} are useful for various operations on strings.

\newcommand{\rzero}{\overline{0}}
\newcommand{\rone}{\overline{1}}
Weighted regular expressions are defined via semirings.
\begin{definition}
Semiring $(R,\oplus, \otimes,\rzero,\rone)$ consists of
set $R$,
two operations $\oplus$ and $\otimes$ on $R$,
two elements $\rzero$ and $\rone$ in $R$,
satisfying the following properties.
\begin{itemize}
\item $\oplus$ and $\otimes$ are associative:
$a \oplus (b\oplus c) = (a \oplus b)\oplus c$ and
$a \otimes (b\otimes c) = (a \otimes b) \otimes c$.
\item $\oplus$ is commutative: $a \oplus b = b \oplus a$.
\item $\rzero$ and $\rone$ are the units of $\oplus$ and $\otimes$, respectively:
$a \oplus\rzero = \rzero \oplus a = a$ and
$a \otimes \rone = \rone \otimes a = a$.
\item $\otimes$ distributes over $\oplus$: namely, $a \otimes (b\oplus c) = (a \otimes b)\oplus (a \otimes c)$ and $(b\oplus c) \otimes a = (b \otimes a)\oplus (c \otimes a)$.
\item $\rzero$ is the annihilator of $\otimes$: namely,
$a \otimes \rzero = \rzero \otimes a = \rzero$. \qed
\end{itemize}
\end{definition}
In weighted regular expressions, we commonly impose additional properties on the semiring.
This paper assumes that $\oplus$ is idempotent, namely $a \oplus a = a$.

A weighted regular expression over semiring $(R,\oplus, \otimes,\rzero,\rone)$
calculates a value in $R$ for every string. Specifically,
$\otimes$ and $\oplus$ are used corresponding to the string concatenation and multiple matching possibilities, respectively.

For studying weighted REwLA, we need to define the computation corresponding to the lookahead.
This paper takes the simplest choice: lookaheads compute nothing regarding the semiring.
This definition can be read as saying that the lookahead does not fetch any character, or that we study weighted regular expressions that may contain lookaheads.

Now we define the syntax of weighted REwLA $\check e$ as follows.
The new syntactic construct $k \in R$ expresses the operation on the weight.
\[\begin{array}{@{}l@{~\,\!}c@{~\,}l@{}}
\check e &::=&
\emps
~\,{\mid}~\,\sigma
~\,{\mid}~\,k
~\,{\mid}~\,\check e \mid \check e
~\,{\mid}~\,\check e \cdot \check e
~\,{\mid}~\,\check e{*}
~\,{\mid}~\,\pla{e}
~\,{\mid}~\,\nla{e}
\end{array}\]
Note that the lookahead is defined by a REwLA.
In what follows, we may write $k\otimes \check e$ and $\check e \otimes k$
instead of $k\cdot \check e$ and $\check e \cdot k$, respectively, to make
$k \in R$ explicit.

Next, we define the value of weighted REwLA matching.
\begin{definition}
The value of weighted REwLA $\check e$ on semiring
$(R,\oplus, \otimes,\rzero,\rone)$ against string $w$, denoted by
$\mathcal{V}(\check e,w)$, is defined as follows.
\[\begin{array}{@{}lcl@{}}
\lefteqn{\textstyle\mathcal{V}(\check e,w)
= \bigoplus\setof{v \mid (\emps,v) \in \mathcal{V}'(\check e,(w, \rone))}}\\
\mathcal{V}'({\emps},\check w) &=& \setof{\check w}\\
\mathcal{V}'(\sigma,(\emps,v))
&=& \emptyset\\
\mathcal{V}'(\sigma,(\sigma' w,v))
&=& \mbf{if}~\sigma=\sigma' ~\mbf{then}~\setof{(w,v)}~\mbf{else}~\emptyset\\
\mathcal{V}'(k,(w,v))
&=& \setof{(w,v \otimes k)}\\
\mathcal{V}'(\check e \mid \check e',\check w) &=& \mathcal{V}'(\check e,\check w) \cup \mathcal{V}'(\check e',\check w)\\
\mathcal{V}'(\check e \cdot \check e',\check w) &=&
\scmp{\check y}{\check x \in \!\mathcal{V}'(\check e,\check w) \wedge \check y \!\in \mathcal{V}'(\check e',\check x)}\\
\mathcal{V}'(\check e{*},\check w) &=&
\setof{\check w} \cup
\scmp{\check y}{\check x \in \mathcal{V}'(\check e,\check w) \wedge \check x \neq \check w \wedge \check y \in \mathcal{V}'(\check e{*},\check x)}\\
\mathcal{V}'(\pla{e},(w,v)) &=&
\mbf{if}~w \in \mathcal{L}(e)~\mbf{then}~\setof{(w,v)}~\mbf{else}~\emptyset\\
\mathcal{V}'(\nla{e},(w,v)) &=&
\mbf{if}~w \not\in \mathcal{L}(e)~\mbf{then}~\setof{(w,v)}~\mbf{else}~\emptyset\qed
\end{array}\]
\end{definition}
The value is undefined if the weighted REwLA contains the repetition consisting of $\emps$ and $k \neq \rone$, such as $k\rept$. We disregard such pathological cases.

\subsection{Submatch Extraction by Weighted Regular Expression}
Weighted regular expressions are useful for various purposes. Here, we consider
submatch extractions  by weighted regular expressions.

To the author's knowledge, Fischer et al. \cite{FiHW10} is the pioneer in using weighted regular expressions for regular expression matching. They expressed matching policies in regular expression matching, such as the longest matching, using weighted regular expressions. However, they did not consider submatching.
The following discussion reformulates the study by Laurikari \cite{Laur01}, which used  finite automata with tagged transitions, from the perspective of weighted regular expressions.

Let $e^\dagger$ be the input regular expression. $e_1 \in e_2$ denotes that $e_1$ is a subexpression of $e_2$.
For simplicity, we assume that all subexpressions in $e^\dagger$ are different.

The semiring considered here is the one over the lists containing elements in $\setof{\ms{L}_e \mid e \in e^\dagger} \cup \setof{\ms{R}_e \mid e \in e^\dagger} \cup \setof{\ms{C}_\sigma \mid \sigma \in \Sigma}$.
Consider a weighted regular expression obtained by replacing
every $\sigma \in e^\dagger$ with $(\sigma \otimes [\ms{C}_\sigma])$
and every submatch-required subexpression $e \in e^\dagger$ with $([\ms{L}_e] \otimes e \otimes [\ms{R}_e])$.
$[\ms{C}_\sigma]$, $[\ms{L}_e]$, and $[\ms{R}_e]$ respectively express
the consumption of $\sigma$ in a transition, and the start and end of the submatching for $e$.
The multiplication operator is the list concatenation, which results in the trace of the matching process including the start and end of the submatching.  The addition operator should pick the desired trace.
Several policies (such as the longest) may exist, but here, we only require that the operators must satisfy the semiring property.

The following sequence may be obtained for string $\sigma_1\cdots \sigma_n$.
\[
[\ms{C}_{\sigma_1},\ldots,\ms{C}_{\sigma_i},\ms{L}_e,\ms{C}_{\sigma_{i\!{+}\!1}},\ldots,\ms{C}_{\sigma_j},\ms{R}_e,\ms{C}_{\sigma_{j\!{+}\!1}},\ldots,\ms{C}_{\sigma_n}]
\]
Note that subsequence $[\ms{C}_{\sigma_{i+1}},\ldots,\ms{C}_{\sigma_j}]$ must not contain $\ms{R}_e$.
This sequence expresses that the submatching for $e$ starts when processing string ${\sigma_1}\cdots{\sigma_i}$ and ends when processing ${\sigma_{i+1}}\cdots{\sigma_j}$; the following ${\sigma_{j+1}}\cdots{\sigma_n}$ is scanned afterwards.
Hence, $\sigma_{i+1}\cdots \sigma_j$ is the submatching for $e$.
In general, the matching result may contain more than one $\ms{L}_e$ and $\ms{R}_e$;
then, every subsequence between the corresponding $\ms{L}_e$ and $\ms{R}_e$ is the submatching for $e$.

In the remaining discussions, we just consider semiring $(R,\oplus, \otimes,\rzero,\rone)$ without fixing a particular matching semantics.

\subsection{Weighted Finite Automata}
We will construct finite automata from REwLA. First, we define weighted NFA.
\begin{definition}
{\em Weighted NFA} $(S,\Sigma,\tau,I,F)$
over semiring $(R,\oplus, \otimes,\rzero,\rone)$ is a five tuple.
$S$ is a finite set of states. $\Sigma$ is the finite alphabet.
$\tau : (S \times \Sigma \times S) \to R$ is the transition function.
$I \subseteq S$ and $F \subseteq S$ are the sets of initial states and final states, respectively.

The {\em state transition} $\typeof{\rho_\mathcal{A}}{\Sigma^* \to 2^{S\times R}}$ of
$\mathcal{A} = (S,\Sigma,\tau,I,F)$ is defined below.
\[\begin{array}{@{}lcl@{}}
  \rho_\mathcal{A}(\emps) &=& \{(s_0,\rone) \mid s_0 \in I \}\\
  \rho_\mathcal{A}(w \cdot \sigma) &=&
\{ (s',r \otimes r') \mid  (s,r) \in \rho_\mathcal{A}(w) \wedge {}
((s,\sigma,s'),r') \in \tau\}
\end{array}\]

The value $\mathcal{V}(\mathcal{A},w)$ of $\mathcal{A}$ for string $w \in \Sigma^*$
is defined below.
\[
\mathcal{V}(\mathcal{A},w) =
\bigoplus\setof{r \mid (s,r) \in \rho_\mathcal{A}(w) \wedge s \in F}
\]

Weighted NFA $\mathcal{A}$ {\em accepts} string $w \in \Sigma^{*}$ if and only if $\mathcal{V}(\mathcal{A},w) \neq \rzero$. The set of strings accepted by $\mathcal{A}$ is denoted by $\mathcal{L}(\mathcal{A})$. \qed
\end{definition}

Weighted NFA is called NFA if its transition function $\tau$ returns only $\rone$ or $\rzero$.
NFA is called DFA if its initial state is unique and the transition function $\tau$ is deterministic,
i.e., $((s,a,s_1),r_1) \in \tau \wedge ((s,a,s_2),r_2)\in \tau$ implies $s_1 = s_2$ and $r_1 = r_2$.

Weighted NFA / DFA enables efficient matching.
\begin{theorem}\label{theorem:nfa}
Given weighted NFA $\mathcal{A}$ on semiring $(R,\oplus, \otimes,\rzero,\rone)$,
its value for string $w$ is obtained in $\mr{O}((m+k)n)$ time and  $\mr{O}(m)$ space,
where $n$ is the length of $w$, $m$ is the number of states of $\mathcal{A}$, and $k$ is the size of the transition function of $\mathcal{A}$, if every element in $R$ is constant-size and both $\otimes$ and $\oplus$ are constant time and space.
\end{theorem}
\begin{proof}\normalfont
Let $\mathcal{A}=(\setof{s_1,\ldots,s_m},\Sigma,\tau,I,F)$.
Consider function $\rho_{s_i}$ defined below.
\[\begin{array}{@{}lcl@{}}
  \rho_{s_i}(\emps) &=& \mbf{if}~{s_i}\in I ~\mbf{then}~\rone~\mbf{else}~\rzero\\
  \rho_{s_i}(w \cdot \sigma) &=&
\bigoplus\setof{\rho_{s_j}(w) \otimes v \mid ((s_j,\sigma,s_i),v) \in \tau}
\end{array}\]
Function $\rho_{s_i}$ can be evaluated by
the standard implementation that repeats updating the vector $(\rho_{s_1},\ldots,\rho_{s_k})$
under the stated computational complexity.
Accordingly, proving the following equation suffices.
\[
\forall s_i.~\forall w \in \Sigma^{*}.~
\bigoplus \setof{r \mid (s_i,r) \in \rho_\mathcal{A}(w)} = \rho_{s_i}(w)
\]
An induction on the length of $w$ proves this equation.
Note that the distributivity and idempotency of $\oplus$ imply
$\bigoplus \setof{r \otimes r' \mid r \in V} = (\bigoplus V) \otimes r'$.
\qed
\end{proof}

\begin{theorem}\label{theorem:dfa}
Given weighted DFA $\mathcal{A}$ on semiring $(R,\oplus, \otimes,\rzero,\rone)$,
its value for string $w$ is obtained in $\mr{O}(m+n)$ time and  $\mr{O}(m)$ space,
where $n$ is the length of $w$ and $m$ is the number of states of $\mathcal{A}$, if every element in $R$ is constant-size and both $\otimes$ and $\oplus$ are constant time and space.
\end{theorem}
\begin{proof}\normalfont
Evaluating $\rho_\mathcal{A}$ exactly as its definition suffices.\qed
\end{proof}
Note that, however, unlike unweighted NFAs, weighted NFAs cannot be determinized or minimized in general. See the literature \cite{DrKV09} for details.

For simplicity of the presentation, we may use $\emps$ transitions, expressed by $((s,\emps,s'),r)$ element in
the transition function $\tau$.
An $\emps$ transition can alter the state without fetching the next character.
We assume that any loop consisting of $\emps$ transitions contains no $((s,\emps,s'),r)$ transition where $r \neq \rone$. This assumption holds for every weighted NFA studied in this paper.

\subsection{Boolean Finite State Automaton}
Boolean Finite State Automaton (BFA)~\cite{BrLe80} (also known as alternating finite automaton~\cite{ChKS81}) is a generalization of DFA and NFA.
Given a string, DFA calculates a state and accepts the string if the state is a final state.
NFA can be regarded as calculating a set of states and checking whether it contains a final state. This view is equivalent to regarding each state as a propositional variable, each of which is assigned True if it is a final state, and calculating their logical disjunction.
BFA may output any propositional formula containing negations and conjunctions.

\begin{definition}
{\em Boolean finite automaton} $(S,\Sigma,\tau,f,F)$ is a five-tuple.
$S$ is a finite set of states.
$\Sigma$ is a finite alphabet.
$\typeof{\tau}{(S \times \Sigma) \to \mathcal{F}_S}$ is a transition function.
$f \in \mathcal{F}_S$ is the initial expression. $F \subseteq S$ is the set of final states.

The {\em state transition} of $\mathcal{A} = (S,\Sigma,\tau,f,F)$, denoted by $\rho_\mathcal{A} : \Sigma^* \to \mathcal{F}_S$, is defined below.
\[\begin{array}{@{}lcl@{}}
  \rho_\mathcal{A}(\emps) &=& f\\
  \rho_\mathcal{A}(w \cdot \sigma) &=& \rho_\mathcal{A}(w)[\tau(s,\sigma)/s]_{s\in S}
\end{array}\]

$\mathcal{A}$ {\em accepts} $w \in \Sigma^*$ if and only if $F \vdash \rho_\mathcal{A}(w)$.
The set of strings accepted by $\mathcal{A}$ is denoted by $\mathcal{L}(\mathcal{A})$.
\qed
\end{definition}
For simplicity, if no $f'$ satisfies $((s,\sigma),f') \in \tau$, we let $\tau(s,\sigma) = \mi{False}$.

Every BFA accepts a regular language.
\begin{theorem}[\fitcite\cite{BrLe80,ChKS81}]\label{theorem:bfa}
For any BFA $\mathcal{A} = (S,\Sigma,\tau,f,F)$,
there exists DFA $\mathcal{B}$ such that
$\mathcal{L}(\mathcal{A}) = \mathcal{L}(\mathcal{B})$
and its number of states is $\mr{O}(2^{2^{|S|}})$.
\end{theorem}
\begin{proof}\normalfont
Construct the NFA whose states consist of every element of $\mathcal{F}_S$. Note that this NFA is a DFA. \qed
\end{proof}

The algorithm to convert BFAs to DFAs is beyond the scope of this paper. We disregard its computational cost. 

BFA supports many operations. Here, we introduce the negation and product.
The negation BFA accepts strings that are not accepted by the original BFA.
The product accepts strings that are accepted by both original BFAs.
\begin{definition}
Given BFA $\mathcal{A} = (S,\Sigma,\tau,f,F)$, its {\em negation} is $\neg\mathcal{A} = (S,\Sigma,\tau,\neg f,F)$. \qed
\end{definition}
\begin{definition}
Given two BFAs $\mathcal{A}_1 = (S_1,\Sigma,\tau_1,f_1,F_1)$
and $\mathcal{A}_2 = (S_2,\Sigma,\tau_2,f_2,F_2)$, their {\em product} is $\mathcal{A}_1 \times \mathcal{A}_2 = (S_1 \cup S_2,\Sigma,\tau_1 \cup \tau_2,f_1 \wedge f_2,F_1 \cup F_2)$.\qed
\end{definition}

\section{Translation of Regular Expressions with Lookaheads into Deterministic Finite Automata}\label{section:DFA}
First, we formulate a method to convert an unweighted REwLA into an automaton by extending the classical NFA construction method known as the Thompson construction \cite{Thom68}.

\begin{figure}
\centering
\includegraphics{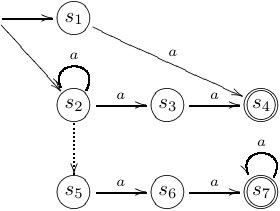}
\caption{The BFA corresponding to $(\pla{aaa{*}}a)\rept a$.
The dotted arrows represent lookaheads.
The initial expression is $s_1 \vee (s_2 \wedge s_5)$, and $\tau(s_2,a) = (s_2 \wedge s_5) \vee s_3$.
The final states are $s_4$ and $s_7$.
}\label{figure:bfa-ex}
\end{figure}

As an example, consider $(\pla{aaa{*}}a)\rept a$.
The automaton corresponding to it
should be something like the one in Figure \ref{figure:bfa-ex}.
Note that states $s_5$ through $s_7$ correspond to the lookahead.
The initial state is $s_1$ or $s_2$. As usual, we would check whether a transition from either of them reaches a final state. Since a lookahead is imposed at state $s_2$, the transitions from $s_2$ and $s_5$  both (i.e., conjunctively) must reach the corresponding final state.

The example indicates that the automaton corresponding to a REwLA can involve dependencies consisting of disjunction, conjunction, and negation. We handle these dependencies using BFA.
Below, we give the conversion $\mi{toBFA}$ from REwLA to BFA.
During this construction, the state names must be suitably renamed so that they are mutually distinct.
In addition, we use the following auxiliary function $\theta^f_F$.
\[\theta^f_F(t) = t[s\vee f/s]_{s\in F}\]
Intuitively, applying $\theta^f_F$ corresponds to adding an $\epsilon$ edge from each state in $F$ to the states expressed by the formula $f$.

\[
\newcommand{\mul}[1]{\multicolumn{3}{@{}l@{}}{#1}}
\begin{array}{@{}lcl@{}}
\lefteqn{\begin{array}{@{}lcl@{}}
\mi{toBFA}(e) &=& (S,\Sigma,\tau,f,F \cup L)
\qquad\mbf{where}~(S,\Sigma,\tau,f,F,L) = \mathcal{T}_B(e)
\end{array}}
\smallskip\\
\mathcal{T}_B(\emps) &=&
(\setof{s},\Sigma,\emptyset,s, \setof{s} ,\emptyset)\\
\mathcal{T}_B(\sigma) &=& (\setof{s,t}, \Sigma,\setof{((s,\sigma),t)},s,\setof{t},\emptyset)
\\
\mathcal{T}_B(e\mid e') &=& (S \cup S', \Sigma, \tau \cup \tau', f \vee f',F\cup F',L \cup L')
\\
\mul{\qquad\begin{array}{@{}ll@{}}
\mbf{where}&(S,\Sigma,\tau,f,F,L) = \mathcal{T}_B(e)\\
&(S',\Sigma,\tau',f',F', L') = \mathcal{T}_B(e')\\
\end{array}}
\smallskip\\
\mathcal{T}_B(e \cdot e') &=& (S \cup S',\Sigma,\tau'',\theta^{f'}_F(f), F', L\cup L')\\
\mul{\qquad\begin{array}{@{}ll@{}}
\mbf{where}&(S,\Sigma,\tau,f,F,L) = \mathcal{T}_B(e)\\
&(S',\Sigma,\tau',f',F', L') = \mathcal{T}_B(e')\\
&\tau''=\setof{((s,\sigma),\theta^{f'}_F(t)) \mid ((s,\sigma),t) \in \tau} \cup \tau'\\
\end{array}}
\smallskip\\
\mathcal{T}_B(e{*}) &=& (S\cup \setof{s_0}, \Sigma,\tau',s_0 \vee f, F\cup \setof{s_0}, L)
\\
\mul{\qquad\begin{array}{@{}ll@{}}
\mbf{where}
&(S,\Sigma,\tau,f,F,L) = \mathcal{T}_B(e)\\
&\tau' = \setof{((s,\sigma),\theta^{f}_F(t)) \mid ((s,\sigma),t) \in \tau}\\
\end{array}}
\smallskip\\
\mathcal{T}_B(\pla{e}) &=& (S\cup\setof{s}, \Sigma, \tau, s \wedge f, \setof{s}, F\cup L)
\\
\mathcal{T}_B(\nla{e}) &=&
(S\cup\setof{s}, \Sigma, \tau, s\wedge\neg f, \setof{s}, F\cup L)
\end{array}
\]

The function $\mathcal{T}_B$ constructs the BFA by a process similar to the Thompson construction \cite{Thom68}.
There are two main differences: it outputs a conjunction for lookaheads, and
it manages two kinds of final states separately.
In the previous example, the two kinds of final states correspond to state $s_4$ and state $s_7$, respectively.
State $s_4$ is the final state of the non-lookahead part, and is affected by concatenation of regular expressions.
On the other hand, state $s_7$ is the final state of the lookahead, and is not affected by the following subexpressions.
Figure \ref{figure:bfa} shows the behavior of $\mathcal{T}_B$ for the case of concatenation.

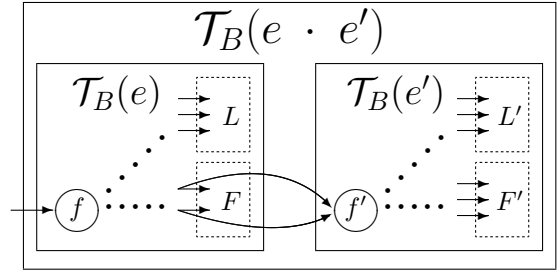
\begin{figure}
\begin{minipage}{\columnwidth}%
\centering
\begin{picture}(205,100)
\put(5,0){\framebox(200,100){}}
\put(105,85){\makebox[0pt][c]{\LARGE $\mathcal{T}_B(e ~{}\cdot{}~ e')$}}

\put(5,-3){
\put(5,10){\framebox(85,70){}}
\put(-5,25){\vector(1,0){16}}
\put(20,25){\circle{16}}\put(20,23){\makebox[0pt][c]{$f$}}

\put(35,65){\makebox[0pt][c]{\Large $\mathcal{T}_B(e)$}}

\put(-3,0){
\put(35,22){\makebox[0pt][c]{\LARGE $\cdot$}}
\put(40,22){\makebox[0pt][c]{\LARGE $\cdot$}}
\put(45,22){\makebox[0pt][c]{\LARGE $\cdot$}}
\put(50,22){\makebox[0pt][c]{\LARGE $\cdot$}}
\put(55,22){\makebox[0pt][c]{\LARGE $\cdot$}}

\put(35,28){\makebox[0pt][c]{\LARGE $\cdot$}}
\put(40,33){\makebox[0pt][c]{\LARGE $\cdot$}}
\put(45,38){\makebox[0pt][c]{\LARGE $\cdot$}}
\put(50,43){\makebox[0pt][c]{\LARGE $\cdot$}}
\put(55,48){\makebox[0pt][c]{\LARGE $\cdot$}}
}

\put(65,47){\dashbox(20,28){}} \put(78,57){\makebox[0pt][c]{$L$}}
\put(58,55){\vector(1,0){12}}
\put(58,61){\vector(1,0){12}}
\put(58,67){\vector(1,0){12}}

\put(65,15){\dashbox(20,28){}} \put(78,25){\makebox[0pt][c]{$F$}}
\put(58,33){\vector(1,0){12}}
\put(58,25){\vector(1,0){12}}

}

\put(5,-3){
\qbezier(58,33)(95,47)(115,27)
\put(110,32){\vector(1,-1){6}}
\qbezier(58,25)(95,13)(115,23)
\put(110,20.5){\vector(2,1){6}}
}

\put(110,-3){
\put(5,10){\framebox(85,70){}}
\put(20,25){\circle{16}}\put(20,23){\makebox[0pt][c]{$f'$}}

\put(35,65){\makebox[0pt][c]{\Large $\mathcal{T}_B(e')$}}

\put(-3,0){
\put(35,22){\makebox[0pt][c]{\LARGE $\cdot$}}
\put(40,22){\makebox[0pt][c]{\LARGE $\cdot$}}
\put(45,22){\makebox[0pt][c]{\LARGE $\cdot$}}
\put(50,22){\makebox[0pt][c]{\LARGE $\cdot$}}
\put(55,22){\makebox[0pt][c]{\LARGE $\cdot$}}

\put(35,28){\makebox[0pt][c]{\LARGE $\cdot$}}
\put(40,33){\makebox[0pt][c]{\LARGE $\cdot$}}
\put(45,38){\makebox[0pt][c]{\LARGE $\cdot$}}
\put(50,43){\makebox[0pt][c]{\LARGE $\cdot$}}
\put(55,48){\makebox[0pt][c]{\LARGE $\cdot$}}
}

\put(65,47){\dashbox(20,28){}} \put(78,57){\makebox[0pt][c]{$L'$}}
\put(58,55){\vector(1,0){12}}
\put(58,61){\vector(1,0){12}}
\put(58,67){\vector(1,0){12}}

\put(65,15){\dashbox(20,28){}} \put(78,25){\makebox[0pt][c]{$F'$}}
\put(58,35){\vector(1,0){12}}
\put(58,29){\vector(1,0){12}}
\put(58,23){\vector(1,0){12}}
}

\end{picture}%
\end{minipage}
\caption{The process of converting REwLA into BFA (the case of concatenation)}\label{figure:bfa}
\end{figure}

This allows us to convert REwLA into BFA, and in turn into DFA.
\begin{theorem}\label{theorem:la-bfa}
\[\mathcal{L}(e) = \mathcal{L}(\mi{toBFA}(e))\]
\end{theorem}
\begin{proof}\normalfont
For $(S, \Sigma, \tau, f, F, L) = \mathcal{T}_B(e)$,
consider function $\tilde\rho$ defined below:
\[
\begin{array}{lcl}
\tilde\rho(f',\emps) &=& f'\\
\tilde\rho(f',w\cdot \sigma) &=& \tilde\rho(f',w)[\tau(s,\sigma)/s]_{s\in S}
\end{array}
\]
Using $\tilde\rho$, we can express the situation that ``with $w_1$ the non-lookahead part reaches a final state, and the lookahead also succeeds using the remaining input $w_2$'' as follows.
\[
L \vdash \tilde\rho(\tilde\rho(f,w_1)[\mi{True}/s]_{s\in F},w_2)
\]
This expression holds if and only if $w_2 \in \mathcal{L}'(e,w_1\cdot w_2)$.
This fact can be easily verified by induction on the structure of REwLA $e$.
Therefore, $L \vdash \tilde\rho(f,w)[\mi{True}/s]_{s\in F}$
holds if and only if $\emps \in \mathcal{L}'(e,w)$.
\qed
\end{proof}
\begin{corollary}
For REwLA $e$, $\mathcal{L}(e)$ is a regular language.\qed
\end{corollary}
\begin{corollary}\label{corollary:dfa-size}
For a REwLA $e$ of length $m$, we can construct a DFA $\mathcal{A}$ with $\mr{O}(2^{2^m})$ states such that $\mathcal{L}(e) = \mathcal{L}(\mathcal{A})$.
\end{corollary}
\begin{proof}\normalfont
This follows from Theorem \ref{theorem:bfa} and the fact that $\mi{toBFA}$ constructs a BFA with $\mr{O}(m)$ states.\qed
\end{proof}
\begin{corollary}
Matching of a REwLA of length $m$ against a string of length $n$ can be computed
in $\mr{O}(n + 2^{2^m})$ time and $\mr{O}(2^{2^m})$ space.
Note that this complexity does not include the cost of constructing the DFA. \qed
\end{corollary}

The resulting DFA may contain a very large number of states, $\mr{O}(2^{2^m})$.
We might be able to reduce the number by DFA minimization in practice.

Reconsider the REwLA $(\pla{aaa{*}}a)\rept a$ as an example.
By $\mi{toBFA}$,
we obtain the BFA in Figure \ref{figure:bfa-ex2}, which is equivalent to that in Figure \ref{figure:bfa-ex} but somewhat more complex.
Note that each state $s_i^j$ in Figure \ref{figure:bfa-ex2} corresponds to state $s_i$ in Figure \ref{figure:bfa-ex}.

\begin{figure}
\begin{minipage}{\columnwidth}
\centering
\includegraphics{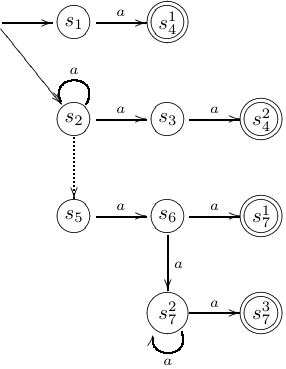}
\end{minipage}
\caption{The BFA obtained as the result of $\mi{toBFA}((\pla{aaa{*}}a)\rept a)$.
The initial expression is $s_1 \vee (s_2 \wedge s_5)$,
$\tau(s_2,a) = (s_2 \wedge s_5) \vee s_3$, and
the final states are $s_4^1$, $s_4^2$, $s_7^1$, $s_7^3$.
}\label{figure:bfa-ex2}
\end{figure}

\section{Translation of Weighted Regular Expressions with Lookaheads into Weighted Finite Automata}\label{section:weight}
In the previous section, we showed a method to convert REwLA into DFA using BFA.
However,
the conversion from BFA to DFA completely destroys the structure of the original expression, and hence,
cannot immediately lead to submatch extraction.
In other words, the conversion from BFA to DFA is nontrivial when there are weights.

Recall $(\pla{aaa{*}}a)\rept a$ and the
BFA corresponding to it (Figure \ref{figure:bfa-ex}).
This BFA retains the structure of the original regular expression, enabling us to assign an appropriate weight to each transition.
In light of this observation, it would be natural to construct a DFA and a weighted NFA, corresponding to the lookahead part and the remaining part, respectively, and run these two in parallel.
However, a simple parallel run is not sufficient.
For example,
consider the state of the BFA in Figure \ref{figure:bfa-ex} after scanning $a$;
the reachable states are $s_2$, $s_3$, $s_4$, $s_5$, $s_6$;
the string is accepted if ``$s_4$ or ($s_3$ and $s_6$) or ($s_2$ and $s_5$ and $s_6$)'' reaches a final state.
Note that the conditions ``$s_5$'' and ``$s_6$'' are required by the lookaheads of the second and first iterations of the repetition, respectively.
This demonstrates that the presence of the repetition may introduce
constraints that appear impossible to derive solely from the automaton corresponding to the lookahead part.

Taking the above discussion into account, we formulate function $\mi{toNFA}$ that constructs a weighted NFA.
It consists of two components.
\[
\mi{toNFA}(\check e) = \mathcal{T}^2_N(\mathcal{T}^1_N(\check e))
\]

First, we define $\mathcal{T}^1_N$. This is exactly the Thompson construction \cite{Thom68},
except that it regards lookaheads as characters.
Note that the transition function returns $\rzero$ for all undefined inputs.
Also, the state names are suitably renamed so that they are mutually distinct.
\newcommand{\plas}[1]{{\lfloor\!{#1}\!\rfloor}}
\newcommand{\nlas}[1]{{\lceil\!{#1}\!\rceil}}
\newcommand{\mul}[1]{\multicolumn{3}{@{}l@{}}{#1}}
\renewcommand{\smallskip}{\relax}
\[
\begin{array}{@{}lcl@{}}
\mathcal{T}^1_N(\emps) &=&
(\setof{s},\Sigma,\emptyset,\setof{s},\setof{s})
\smallskip\\
\mathcal{T}^1_N(\sigma) &=& (\setof{s,t},\Sigma, \setof{((s,\sigma, t),\rone)},\setof{s},\setof{t})
\smallskip\\
\mathcal{T}^1_N(k) &=& (\setof{s,t},\Sigma, \setof{((s,\emps, t),k)},\setof{s},\setof{t})
\smallskip\\
\mathcal{T}^1_N(\check e_1 \mid \check e_2) &=& (S_1 \cup S_2 \cup \setof{\hat s,\hat t}, \Sigma_1\! \cup \Sigma_2, \hat\tau, \setof{\hat s}, \setof{\hat t})\\
\mul{\qquad\begin{array}{@{}ll@{}}
\mbf{where}&
(S_1,\Sigma_1,\tau_1,\setof{s_1},\setof{t_1}) = \mathcal{T}^1_N(\check e_1)\\
&(S_2,\Sigma_2,\tau_2,\setof{s_2},\setof{t_2}) = \mathcal{T}^1_N(\check e_2)\\
&\hat\tau = \tau_1 \cup \tau_2 \cup
\{((\hat s,\emps,s_1),\rone),((\hat s,\emps,s_2),\rone), ((t_1,\emps,\hat t),\rone),((t_2,\emps,\hat t),\rone)\}
\end{array}}
\smallskip\\
\mathcal{T}^1_N(\check e_1 \cdot \check e_2) &=& (S_1 \cup S_2, \Sigma_1 \cup \Sigma_2, \hat\tau, I_1, F_2)\\
\mul{\qquad\begin{array}{@{}ll@{}}
\mbf{where}&
(S_1,\Sigma_1,\tau_1,I_1,\setof{t_1}) = \mathcal{T}^1_N(\check e_1)\\
&(S_2, \Sigma_2,\tau_2,\setof{s_2}, F_2) = \mathcal{T}^1_N(\check e_2)\\
&\hat\tau = \tau_1\cup \tau_2 \cup\setof{((t_1,\emps, s_2),\rone)}
\end{array}}
\smallskip\\
\mathcal{T}^1_N(\check e{*}) &=& (S,\Sigma', \hat \tau,\setof{s},\setof{t})\\
\mul{\qquad\begin{array}{@{}ll@{}}
\mbf{where}&
(S,\Sigma',\tau,\setof{s},\setof{t}) = \mathcal{T}^1_N(\check e)\\
& \hat\tau = \tau \cup \{((s,\emps,t),\rone),((t,\emps,s),\rone)\}\\
\end{array}}
\smallskip\\
\mathcal{T}^1_N(\pla{e}) &=&
(\setof{s,t},\Sigma \cup \setof{\plas{e}}, \setof{((s,\plas{e},t),\rone)},\setof{s},\setof{t})
\smallskip\\
\mathcal{T}^1_N(\nla{e}) &=&
(\setof{s,t},\Sigma \cup\setof{\nlas{e}}, \setof{((s,\nlas{e},t),\rone)},\setof{s},\setof{t})
\end{array}\]

The function $\mathcal{T}^2_N$ constructs a BFA representing the constraints regarding lookaheads and takes the product with the original automaton.
This BFA adds, for each lookahead, the corresponding initial expression as a constraint, and changes the constraints according to the characters read.

\[
\begin{array}{@{}lcl@{}}
\mathcal{T}^2_N(S,\Sigma',\tau,I,F) &=& (S \!\times\! \mathcal{F}_{B}, \Sigma,\hat\tau, I \!\times\! \setof{\mi{True}}, \scmp{(s,f)}{s \!\in\! F \wedge F_B \!\vdash\! f})\\
\mul{\qquad\begin{array}{@{}ll@{}}
\mbf{where}&
\mathcal{S} =
\scmp{(\plas{e},\mi{toBFA}(e))}{\plas{e} \in \Sigma' \setminus \Sigma}
{}\cup{}
\scmp{(\nlas{e},\neg\mi{toBFA}(e))}{\nlas{e} \in \Sigma' \setminus \Sigma}\\
&
(B,\Sigma,f_B,\tau_B, F_B) = \prod_{(e,\mathcal{B}) \in \mathcal{S}} \mathcal{B}\\
&
\hat\tau = \{(((s,f),\delta(e),(t,\xi(e,f))),r) \mid {}
((s,e,t),r) \in \tau\wedge f \in \mathcal{F}_B \} \\
& \delta(e) = \mathbf{if}~e \in \Sigma ~\mbf{then}~e~\mbf{else}~\emps\\
& \xi(e,f) = \mbf{if}~e \in \Sigma ~\mbf{then}~f[\tau_B(e,b)/b]_{b\in B}\\
&\phantom{\xi(e,f) ={}}
\mbf{else}~\mbf{if}~e = \emps ~\mbf{then}~f\\
&\phantom{\xi(e,f) ={}}
\mbf{else}~\mbf{let}~(\tilde B,\Sigma,\tilde \tau,\tilde f,\tilde F) = \mathcal{S}(e)~\mbf{in}~f\wedge \tilde f
\end{array}}
\end{array}
\]

By the above construction, we can convert a weighted REwLA into a weighted NFA.
\begin{lemma}\label{lemma:weighted-1}
For $(S,\Sigma',\tau,\setof{s_0},\setof{t_0}) = \mathcal{T}^1_N(\check e)$, consider the minimal $\rho^\dagger$ satisfying the following.
\[\begin{array}{@{}l@{~\,}c@{~\,}l@{}}
  \rho^\dagger(\emps,w') &\supseteq& \{(s_0,\rone)\}\\
  \rho^\dagger(w\sigma,w') &\supseteq&
\{ (t,r \otimes r')\mid (s,r) \in \rho^\dagger(w,\sigma w') \wedge {}
((s,\sigma,t),r') \in \tau\}\\
  \rho^\dagger(w,w') &\supseteq& \\
\multicolumn{3}{@{}l@{}}{\quad
\{(t,r \otimes r') ~\mid~
(s,r) \in \rho^\dagger(w,w') \wedge ((s,a,t),r') \in \tau \wedge {}
( a = \emps \vee (a = \plas{e} \wedge w' \in \mathcal{L}(e))\vee {}
(a = \nlas{e} \wedge w' \not\in \mathcal{L}(e)))\}
}
\end{array}
\]
Then, the following holds.
\[
\begin{array}{@{}l@{}}
(w',v) \in \check{\mathcal{V}}(\check e,(w\cdot w',\rone))
\iff
(t_0,v) \in \rho^\dagger(w,w')
\end{array}
\]
\end{lemma}
\begin{proof}\normalfont
This can be easily shown by induction on the structure of $\check e$.\qed
\end{proof}
\begin{lemma}\label{lemma:weighted-2}
For $(S,\Sigma',\tau,\setof{s_0},\setof{t_0}) = \mathcal{T}^1_N(\check e)$,
consider the $\rho^\dagger$ defined in Lemma \ref{lemma:weighted-1}.
Let $\mathcal{A} = \mathcal{T}^2_N(\mathcal{T}^1_N(\check e))$ and let $F$ be the set of final states of $\mathcal{A}$; then the following holds.
\[
(\exists (t,f)\in F.~ ((t,f),v) \in \rho_{\mathcal{A}}(w))
\iff (t_0,v) \in \rho^\dagger(w,\emps)
\]
\end{lemma}
\begin{proof}\normalfont
We show that the right-hand side implies the left-hand side. The converse can be shown similarly.

Given state transition
$[(s_0,\rone),(s_1,v_1),\ldots,(s_k,v_k),(t_0,v)]$ leading to $(t_0,v)$ realized by $\rho^\dagger$,
there must exist a corresponding transition
$[((s_0,\mi{True}),\rone),((s_1,f_1),v_1),\ldots,((s_k,f_k),v_k),((t_0,f),v)]$ by $\rho_\mathcal{A}$.
Therefore,
it suffices to show $F_B\vdash f$ for $F_B$ defined in $\mathcal{T}^2_N$.
If the state transition from $(s_i,v_i)$ to $(s_{i+1},v_{i+1})$ in $\rho^\dagger$ is labeled with $\plas{e}$ or $\nlas{e}$ (here, tentatively assume $\plas{e}$),
the corresponding transition in $\rho_\mathcal{A}$ is the one from $((s_i,f_i),v_i)$ to $((s_{i+1},f_i\wedge f'),v_{i+1})$, where $f'$ is the initial formula of the BFA corresponding to $e$.
In the other cases the formula is rewritten according to the state transition of the BFA.
Since the transition was possible in $\rho^\dagger$,
the remaining string that has not yet been processed must be contained in $\mathcal{L}(e)$. As stated in Theorem \ref{theorem:la-bfa}, this is equivalent to the transitions of the corresponding BFA satisfying the acceptance condition.\qed
\end{proof}
\begin{theorem}\label{theorem:weighted}
\[\mathcal{V}(\check e,w) = \mathcal{V}(\mi{toNFA}(\check e),w)\]
\end{theorem}
\begin{proof}\normalfont
Apparent from Lemmas \ref{lemma:weighted-1} and \ref{lemma:weighted-2}.
\qed
\end{proof}
\begin{corollary}\label{coro:wnfa-size}
For a weighted REwLA $\check e$ of length $m$, we can construct a weighted NFA $\mathcal{A}$ with $\mr{O}(2^{2^m})$ states such that $\mathcal{V}(\check e,w) = \mathcal{V}(\mathcal{A},w)$.
\end{corollary}
\begin{proof}\normalfont
Let $p$ be the total size of the non-lookahead parts in $\check e$ and $q$ the size of the lookahead parts; then the number of states of the weighted NFA obtained by $\mi{toNFA}$ is
$\mr{O}(p \cdot 2^{2^{q}})$.
This is because $\mathcal{T}^2_N$ uses, as the states of the NFA, formulas whose propositional variables are the states of a BFA with $\mr{O}(q)$ states.
Also, the correctness of this construction follows from Theorem \ref{theorem:weighted}.\qed
\end{proof}
\begin{corollary}\label{coro:wnfa-cost}
The weight computation of a weighted REwLA $\check e$ of length $m$ for a string of length $n$ can be computed
in $\mr{O}(n \cdot 2^{2^m})$ time and $\mr{O}(2^{2^m})$ space, provided that each operator of the semiring can be computed in constant time and constant space.
\end{corollary}
\begin{proof}\normalfont
This follows from Theorem \ref{theorem:nfa} and Corollary \ref{coro:wnfa-size}.
Note that the size of the transition function of the weighted NFA obtained by $\mi{toNFA}$ is
$\mr{O}(p \cdot 2^{2^{q}})$, where $p$ is the total size of the non-lookahead parts in $\check e$ and $q$ is the size of the lookahead parts.
Note that this complexity does not include the time required to construct the weighted NFA.
\qed
\end{proof}

Although the worst-case complexity is doubly exponential in $m$,  the length of the regular expression,
we expect realistic performance in practice.
As seen in Corollary \ref{coro:wnfa-size}, the double exponential blow-up is with respect to the size of the lookaheads; hence, it is not problematic for small lookaheads.
Moreover, by identifying formulas that are not distinguished by transitions,
we may be able to reduce the size of the weighted NFA.

As an example, consider the
REwLA $(a\pla{aa{*}})\rept a$.
For simplicity, we ignore weights.
Applying $\mathcal{T}^1_N$ to this REwLA
yields the NFA in Figure \ref{figure:weight-A1}, and
further applying $\mathcal{T}^2_N$ yields
Figure \ref{figure:weight-example}.
For clarity,
as the BFA obtained from $\mi{toBFA}(aa{*})$,
we use the simpler
BFA $(\setof{b_0,b_1},\scmp{((b,a),b_1)}{b \in \setof{b_0,b_1}},b_0,\setof{b_1})$, which is equivalent to the one computed according to the definition.

\begin{figure}
\begin{minipage}{\columnwidth}
\centering
\begin{picture}(170,50)
\put(-5,-20){
\put(-5,40){\vector(1,0){8}}
\put(10,40){\circle{14}}\put(10,38){\makebox[0pt][c]{$s_1$}}
\put(20,40){\vector(1,0){25}} \put(30,43){$a$}
\put(55,40){\circle{14}}\put(55,38){\makebox[0pt][c]{$s_2$}}
\put(65,40){\vector(1,0){35}} \put(70,43){\small$\plas{aa{*}}$}
\put(110,40){\circle{14}}\put(110,38){\makebox[0pt][c]{$s_3$}}
\put(120,40){\vector(1,0){25}} \put(130,43){$a$}
\put(155,40){\circle{12}}\put(155,38){\makebox[0pt][c]{$s_4$}}
\put(155,40){\circle{16}}

\put(110,50){\line(0,1){10}}
\put(110,60){\line(-1,0){100}}
\put(10,60){\vector(0,-1){10}}
\put(55,62){$\emps$}
\put(110,20){\vector(0,1){10}}
\put(110,20){\line(-1,0){100}}
\put(10,30){\line(0,-1){10}}
\put(55,22){$\emps$}
}
\end{picture}
\end{minipage}%
\caption{The NFA obtained as the result of $\mathcal{T}^1_N((a\pla{aa{*}})\rept a)$}\label{figure:weight-A1}
\end{figure}
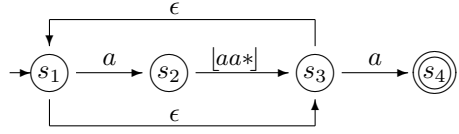

\begin{figure}
\begin{minipage}{\columnwidth}
\centering
\begin{picture}(205,115)
\put(20,30){
\put(0,30){
\put(120,40){\circle{20}}
\put(120,42.5){\makebox[0pt][c]{\small$s_3$}}
\put(120,34.5){\makebox[0pt][c]{\scriptsize$\mi{True}$}}
\put(135,40){\vector(1,0){25}} \put(145,42){$a$}
\put(175,40){\circle{16}}\put(175,40){\circle{20}}
\put(175,42.5){\makebox[0pt][c]{\small$s_4$}}
\put(175,34.5){\makebox[0pt][c]{\scriptsize$\mi{True}$}}
}

\put(65,57){\vector(4,1){44}} \put(60,59){$\emps$}
\put(65,57){\vector(-4,-1){44}}

\put(-8,40){\vector(1,0){8}}
\put(10,40){\circle{20}}
\put(10,42.5){\makebox[0pt][c]{\small$s_1$}}
\put(10,34.5){\makebox[0pt][c]{\scriptsize$\mi{True}$}}
\put(25,40){\vector(1,0){25}} \put(35,42){$a$}
\put(65,40){\circle{20}}
\put(65,42.5){\makebox[0pt][c]{\small$s_2$}}
\put(65,34.5){\makebox[0pt][c]{\scriptsize$\mi{True}$}}
\put(80,40){\vector(1,0){25}} \put(90,42){$\emps$}
\put(120,40){\circle{20}}
\put(120,42.5){\makebox[0pt][c]{\small$s_3$}}
\put(120,34.5){\makebox[0pt][c]{\scriptsize$b_0$}}
\put(135,40){\vector(1,0){25}} \put(145,42){$a$}
\put(175,40){\circle{16}}\put(175,40){\circle{20}}
\put(175,42.5){\makebox[0pt][c]{\small$s_4$}}
\put(175,34.5){\makebox[0pt][c]{\scriptsize$b_1$}}

\put(0,-30){
\put(65,57){\vector(4,1){44}} \put(40,53){$\emps$}
\put(65,57){\vector(-4,-1){44}}

\put(10,40){\circle{20}}
\put(10,42.5){\makebox[0pt][c]{\small$s_1$}}
\put(10,34.5){\makebox[0pt][c]{\scriptsize$b_0$}}
\put(25,40){\vector(1,0){25}} \put(35,42){$a$}
\put(65,40){\circle{20}}
\put(65,42.5){\makebox[0pt][c]{\small$s_2$}}
\put(65,34.5){\makebox[0pt][c]{\scriptsize$b_1$}}
\put(80,40){\vector(1,0){25}} \put(90,42){$\emps$}
\put(120,40){\circle{20}}
\put(120,42.5){\makebox[0pt][c]{\small$s_3$}}
\put(120,34.5){\makebox[0pt][c]{\scriptsize$b_0{\wedge}b_1$}}
\put(135,40){\vector(1,1){25}} \put(140,52){$a$}
}

\put(0,-60){
\put(120,38){\vector(0,1){20}}
\put(120,38){\vector(-1,0){93}}
\put(70,40){$\emps$}
\put(10,40){\circle{20}}
\put(10,42.5){\makebox[0pt][c]{\small$s_1$}}
\put(10,34.5){\makebox[0pt][c]{\scriptsize$b_0{\wedge}b_1$}}
\put(25,40){\vector(1,1){25}} \put(30,52){$a$}
}
}
\end{picture}
\end{minipage}%
\caption{The NFA obtained as the result of $\mi{toNFA}((a\pla{aa{*}})\rept a)$}\label{figure:weight-example}
\end{figure}
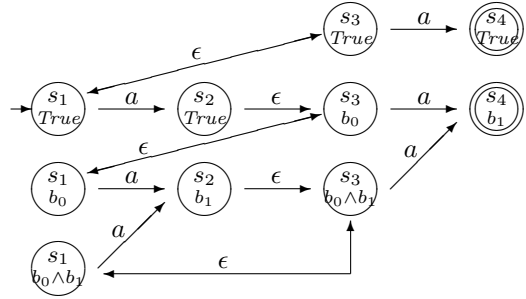

\section{Discussion}
\subsection{Extensions of Regular Expressions and the Complexity of Their Matching}
We have discussed a REwLA matching algorithm.
Well-known variants of regular expressions other than REwLA include Extended Regular Expression and Regular Expression with Back References (abbreviated to ERE and REwBR, respectively).

An ERE is a regular expression that includes operators denoting negation and intersection. It  is similar to regular expressions with lookaheads and negative lookaheads.
There have been several investigations on ERE \cite{Pete02,RoVi03,IlSY03,Rosu07,OwRT09}.
Notably, the number of states of the DFA corresponding to an ERE
cannot be bounded by an elementary function of the size of the ERE \cite{StoM73,Rosu07}.
Therefore, its state-transition-machine-based implementation would be slow, at least in the worst case.
In fact, no matching algorithm whose running time is linear in the length of the input string and elementary in the size of the ERE is known\footnote{Rosu and Viswanathan \cite{RoVi03} proposed a matching algorithm running in $\mr{O}(n \cdot 2^{2m^2})$ time. However, the same authors later stated that this complexity estimate must be wrong \cite{Rosu07}.}.

Unlike ERE and REwLA, the language denoted by REwBR is not necessarily regular.
This might be the reason why its properties are not well understood \cite{Lars98,CaSY03,CaSa09}.
For the string matching problem, the complexity is known to be NP-complete \cite{Aho91}, and
an efficient algorithm seems impossible.

In this study, we have shown that REwLA matching can be computed in $\mr{O}(n \cdot 2^{2^m})$ time.
This is better than the known results for ERE and REwBR, and
suggests that REwLA is a useful extension of regular expressions.
The technically noteworthy point compared to ERE is that, by going through BFA,
we can properly handle even REwLA containing negative lookaheads.
Analyzing why REwLA is easier to handle than ERE is future work.

The most difficult part in implementing the proposed method is deciding the equivalence of propositional formulas.
A promising approach is to adopt binary decision diagrams.
Binary decision diagrams transform a formula into a normal form, and at the same time,
dynamically simplify the formula.
Therefore, we expect that they provide a natural implementation for the BFA states. This implementation might improve the complexity.

\subsection{Realizing Submatch Extraction with Regular Expressions}
To the best of the authors' knowledge,
Dub\'e and Feeley~\cite{DuFe00} are the first to point out the difficulty of obtaining submatches. They gave an algorithm to parse strings based on regular expressions. Laurikari~\cite{Laur01} discussed this problem in detail and proposed an implementation based on an automaton equipped with tags corresponding to the start and end of submatches and an ordering relation over them.
Theoretically, these methods can be formulated as constructions of weighted automata.
Cox \cite{Cox07} summarized the history and implementation of methods that perform regular expression matching without backtracking.

There are several methods for constructing weighted automata from weighted regular expressions:
Caron and Flouret~\cite{CaFl00} based on the Glushkov construction \cite{Glus61},
Champarnaud and Duchamp~\cite{ChDu04} based on the Brzozowski derivative \cite{Brzo64},
and Lombardy and Sakarovitch~\cite{LoSa05} based on Antimirov's partial derivative \cite{Anti96}.
Our approach is based on the Thompson construction \cite{Thom68}, but these methods also apply.

\subsection{On the Handling of Lookaheads}
A tree transducer \cite{FuVo98} is a tree automaton extended with output.
Tree transducers with regular lookahead have been studied in depth.
It is known that macro tree transducers with regular lookahead \cite{EnVo85} can be converted into equivalent macro tree transducers without lookahead.
Since the domain of a macro tree transducer is a regular tree language,
this fact also serves as an indirect proof that the language expressed by REwLA is regular.
Therefore, the observation itself that REwLA is regular is not new.
Note that the results on macro tree transducers do not use BFA; hence, the proof presented in this paper is entirely new.

By processing the input string from the end to the beginning,
the matching of regular expressions containing lookaheads can be computed easily.
First, the lookahead parts and the non-lookahead parts are each
converted into automata that recognize the reversals of the strings matching those subexpressions.
Next, the input string is scanned in reverse and the transitions of each automaton are computed in parallel. At this time, in the parts where a lookahead was required, a transition is allowed only when the automaton corresponding to that lookahead is in an accepting state.
Unfortunately, this method cannot handle lookbehinds and is also not suited to stream processing.
It is worth noting that this method can extend the results of this paper to those including lookbehinds.

Let us develop the idea mentioned above further. For each lookahead and lookbehind, record whether it succeeds at each character as a preprocessing step. Then matching of REwLA, including lookbehinds, is straightforward.
This approach is unsuitable for stream processing and requires $\mr{O}(n)$ space in addition to the input string.
The known results on tree transducers indicate that the preprocessing and the matching calculation can be merged into a single forward scan.
However, it is unclear whether this merging improves the efficiency.
The automaton corresponding to nested negative lookaheads and negative lookbehinds may be very large, if we naively apply determinization for each negation.
Since this situation may arise in the merging process, the complexity estimation requires a detailed examination.

In this paper, we assumed that the lookahead part does not contain the weight computation.
One may hope to extract the part matched by a lookahead. However, if we consider submatching inside the lookahead part, it is nontrivial to provide an appropriate semantics of submatches for a REwLA within a negative lookahead. Further investigation is future work.

{\bf Acknowledgement}\ 
The author is grateful to Kazuhiro Inaba (Google) for his discussions on the complexities and existing studies on translating REwLA to DFAs, Yasuhiko Minamide (Tsukuba University) for his instruction on strategies for regular expression matching, and the anonymous reviewers for their comments, which were useful for improving the presentation.

\section*{Note Added in Translation}
This article is author's unofficial English translation of: A. Morihata,  ``Sakiyomi tsuki seikihyougen no yuugenjoutai ohtomaton heno henkan'', (Computer Software, Vol. 29, No. 1, pp. 147--158, 2012, in Japanese. \doi{10.11309/jssst.29.1_147}). 
The technical content is the same as the original, except for the correction of the following typological mistakes.
\begin{itemize}
\item In the definition of $\mathcal{T}_B(e \cdot e')$, $\mathcal{T}_B$ must process $e'$ recursively.
\item In Lemma~\ref{lemma:weighted-2}, Lemma~\ref{lemma:weighted-1} must be referred.
\item In Corollary~\ref{coro:wnfa-cost}, Corollary~\ref{coro:wnfa-size} must be referred.
\item A few misspellings of people's names are fixed.
\end{itemize}
 The copyright of this material is retained by the Japan Society for Software Science and Technology (JSSST). This material is published with the agreement of the JSSST. Please comply with Copyright Law of Japan if any users wish to reproduce, make derivative work, distribute or make available to the public any part or whole thereof.

Those who are interested in this topic are also referred to the follow-up studies, including, but not limited to, Miyazaki and Minamide (JIP 2019, JIP 2023), Berglund et al. (JUCS 2021), and Mamouras and Chattopadhyay (POPL 2024).
\end{document}